\documentclass[]{article}

\usepackage{amsmath,amssymb,amsfonts}
\usepackage{graphicx}
\usepackage{color}
\usepackage{appendix}
\usepackage{authblk}
\usepackage{amsthm}
\usepackage{tikz}

\usepackage[english]{babel}
\usepackage[utf8]{inputenc}

\usepackage{cite}

\usepackage{geometry}
\newtheorem{theorem}{Theorem}

\newtheorem{definition}{Definition}

\newcommand{\nn}{\nonumber}

\def\a{\alpha}

\def\d{\delta}
\def\eps{\epsilon}

\def\th{\theta}

\def\m{\mu}
\def\n{\nu}

\def\r{\rho}

\def\s{\sigma}

\def\p{\psi}
\def\w{\omega}

\def\P{\Psi}

\def\<{\langle}
\def\>{\rangle}

\def\Tr{{\rm Tr}}

\def\mN{{\mathcal{N}}}
\def\mH{{\mathcal{H}}}
\def\bI{{\mathbb{I}}}
\def\hP{{\hat{\Psi}}} 
\def\hPh{{\hat{\Phi}}}
\def\he{{\hat{e}}}

\def\ha{{\hat{a}}}

\begin{document}
	
	\title{Reduced density matrix of nonlocal identical particles}
	
	\author[1]{Seungbeom Chin \thanks{sbthesy@gmail.com}}
\author[2,3]{Joonsuk Huh \thanks{joonsukhuh@gmail.com}}
\affil[1]{Department of Electrical and Computer Engineering, Sungkyunkwan University, Suwon 16419, Korea}

\affil[2]{Department of Chemistry, Sungkyunkwan University, Suwon 16419, Korea }

\affil[3]{ SKKU Advanced Institute of Nanotechnology (SAINT), Sungkyunkwan University, Suwon 16419, Korea}
\maketitle

	\begin{abstract}
	We probe the theoretical connection among three different approaches to analyze the entanglement of identical particles,  i.e., the first quantization language (1QL), elementary-symmetric/exterior products (which has the mathematical equivalence to no-labeling approaches), and the algebraic approach based on the GNS construction. Among several methods to quantify the entanglement of identical particles, we focus on the computation of reduced density matrices, which can be achieved by the concept of \emph{symmetrized partial trace} defined in 1QL. We show that the symmetrized partial trace corresponds to the interior product in symmetric and exterior algebra (SEA), which also corresponds to the subalgebra restriction in the algebraic approach based on GNS representation. Our research bridges different viewpoints for understanding the quantum correlation of identical particles in a consistent manner.
\end{abstract}

\section{Introduction}

 Quantum entanglement is one of the crucial quantum concepts which reveals the essential feature of quantum physics. It implies the possibility of composite systems that cannot be described as a simple collection of individual subsystems,  even when the subsystems locate far from each other \cite{einstein1935can, bell2004speakable}. It is also exploited as a crucial resource to enable several tasks with quantum speedup \cite{horodecki2009quantum}.
	One of the insightful approaches to analyze the entanglement of a given quantum system is to use the partial trace technic. Intuitively, if a multipartite quantum system is entangled, a subsystem of the quantum system has some nonlocal (measurement-dependent) correlation with the ``outer world'' in the total system. The information on the correlation is encoded in the \emph{reduced density matrix}, a quantum state acquired by partial tracing the outer part of the total state.
	
	For the case of non-identical particles, in which each particle resides in a distinguished Hilbert space, the concept of the partial trace is well-defined.
	On the other hand, in the case of identical particles (both bosons and fermions), individual particles do not reside in independent Hilbert spaces, by which the concept of partial trace seems inappropriate to obtain reduced density matrices suitable for analyzing the entanglement of identical particles. Several alternatives have been suggested to overcome this problem.

	An algebraic approach (AA) based on Gel'fand-Naimark-Segel (GNS) construction is suggested by Balachandran et al.~\cite{balachandran2013entanglement,balachandran2013algebraic}.
	They demonstrated that, instead of partial trace, a restriction to a chosen subalgebra of observables provides a sound tool to compute the entanglement entropy (the importance of subalgebras for interpreting the entanglement of identical particles in the second quantization language (2QL) is also pointed out in Ref.~\cite{benatti2011entanglement,benatti2012bipartite,benatti2017remarks}). On the other hand, a no-labelling approach (NLA) \cite{franco2016quantum, franco2018indistinguishability,compagno2018dealing} describes the identical particles without introducing particle pseudo-labels. By extracting the transition relations of wave functions from the symmetric properties of identical particles, partial trace can be defined in the no-labeling formalism (the same process was reproduced in the second quantization language in Ref.~\cite{lourencco2019entanglement}). Also, the concept of \emph{symmetrized partial trace} for identical particles in the first quantization language (1QL) was introduced in Ref.~\cite{chin2019entanglement}. By applying the symmetrization principle for identical particles not only to wavefunctions but also to detectors, the reduced density matrix that preserves the particle label symmetry can be derived (for more discussions on the entanglement of identical particles from various viewpoints, see Refs. \cite{ schliemann2001j,schliemann2001ja, eckert2002k, paskauskas2001r,paunkovic2004role, zanardi2002p,shi2003shi,barnum2004subsystem, barnum2005generalization, levay2005elementary, plastino2009ar, killoran2014extracting,tichy2013entanglement,dalton2017quantum,dalton2017quantum2, cavalcanti2007useful}).
	
	Hence, one can state that there exists a seemingly incompatible standpoint on the validity of the partial trace in the problem of entanglement of identical particles. However, considering that all the above theories deal with the same physical systems, one can surmise that this seeming inconsistency must be reconciled by examining the resultant quantities or the formal structures.
	In this work, we show how the aforementioned different approaches can translate and provide fundamentally equivalent results to each other.
	
	In the course of delving into the issue, we utilize the elementary symmetric product (for bosons) and the exterior product (for fermions) to demonstrate the indistinguishability and projection rules for identical particles. It will be shown that \emph{the symmetrized partial trace of identical particles is equivalent to the interior products between elementary symmetric (for bosons) and  exterior (for fermions) state vectors, which is equivalent to the restriction to a subalgebra on GNS representation.}
	
	 Section~\ref{1ql} explains the first quantization language (1QL) and the concept of symmetrized partial trace. Section~\ref{sea} shows how 1QL can be translated to the symmetric and exterior algebra (SEA). Then Section~\ref{algebraic} shows that reduced density matrix computed in SEA is the same matrix computed using the concept of restriction to subalgebras in the algebraic approach introduced in Refs.~\cite{balachandran2013entanglement,balachandran2013algebraic}.  
	
	
\section{The first quantization language}\label{1ql}


To describe the identical particles in 1QL, we start with the description of $N$ non-identical particles. For this case, since we can ``distinguish'' the particles fully, each particle has inherent physical labels that are different from each other. A particle with label A in a state $\P$ is described by $|\P\>_A$,  where
$\P$ includes the position $\p$ and internal state $s$, i.e., $\P = (\p,s)$. The total $N$-particle state is given by
\begin{align}\label{nnonid}
    |\P_1\>_{1}\otimes|\P_2\>_{2}\otimes\cdots \otimes |\P_N\>_{N}, 
\end{align} where the subscripts outside the kets denote the particle labels.
Since the transition amplitude of $|\P\>_A$ to another state $|\Phi\>_B$ is given by
\begin{align}\label{trans}
_A\<\Phi|\P\>_B = \<\Phi|\P\> \d_{AB},
\end{align} 
the transition amplitude from a $N$-particle state $ |\P_1\>_{1}\otimes|\P_2\>_{2}\otimes\cdots \otimes|\P_N\>_{N}$ to $ |\Phi_1\>_{1}\otimes|\Phi_2\>_{2}\otimes\cdots \otimes|\Phi_N\>_{N}$ is computed as
\begin{align}\label{nnonidtrans}
     \<\Phi_1|\P_1\>\<\Phi_2|\P_2\>\cdots \<\Phi_N|\P_N\>.
\end{align}
Or, one can omit the labels, by  supposing that the order of particle states represents the labels. Then Eq.~\eqref{nnonid} is rewritten as
\begin{align}\label{nnonidordered}
|\P_1\>\otimes|\P_2\>\otimes\cdots \otimes|\P_N\> 
\end{align}
and the transition amplitude is given in the same form as Eq.~\eqref{nnonidtrans}. We call the  expressions Eq.~\eqref{nnonid} and Eq.~\eqref{nnonidordered}  \emph{explicit/implicit notation} respectively.  However, one should be aware that a ket $|\P_i\>$ in the implicit notation is different from the same ket in Eq.~\eqref{nnonidtrans}. The kets in Eq.~\eqref{nnonidordered} implicitly contain the information on particle labels by their relative order, which is absent in the kets in Eq.~\eqref{nnonidtrans}. Hence, we use the explicit notation to avoid the confusion in this text. Appendix \ref{implicit} explains the implicit description of identical particles and compares it with the explicit description that we will discuss from now on.

Since identical particles have the exchange symmetry, 
the total state of $N$ identical particles in states $\P_i$ ($i=1,\dots,N$) with particle labels $A_a$ ($a=1,\dots, N$ and $A_a\neq A_b$ for $a\neq b$) is expressed as  
\begin{align}\label{totalstate}
|\P\>^\pm \equiv 
 |\P_1,\P_2, \cdots , \P_N\>^{\pm} = \mN(\P)\Big[ \sum_{\s\in S_N}(\pm 1)^\s|\P_{1}\>_{A_{\s(1)}}|\P_{2}\>_{A_{\s(2)}}\cdots |\P_{N}\>_{A_{\s(N)}}\Big],
\end{align}
where the sign $+$ is for bosons and $-$ for fermions. $\mN(\P)$ is the normalization factor of $ |\P_1,\P_2, \cdots , \P_N\>$. Since identical particles cannot be addressed individually, the particle labels are not physical, i.e., ``pseudo-labels''~\cite{chin2019entanglement}.
With Eq.~\eqref{trans}, the transition amplitude from $|\P_1,\cdots , \P_N\>^\pm$ to $|\Phi_1,\cdots , \Phi_N\>^\pm$ is computed as 
\begin{align}\label{ta}
\< \Phi_1,\cdots , \Phi_N|\P_1,\cdots , \P_N\>^\pm = \mN(\Phi)\mN(\P) \sum_{\r,\s}(\pm 1)^{\r}(\pm 1)^\s \prod_{i=1}^{N}\<\Phi_{\r(i)}|\P_{\s(i)}\>.
\end{align}
By the definition of a matrix $A$ with imposed as $A_{ij} =\<\Phi_i|\P_j\> $, Eq. \eqref{ta} is proportional to the matrix permanent and determinant of $A$ for bosons and fermions repectively.


\begin{figure}[t]
	\centering
	\includegraphics[width=8cm]{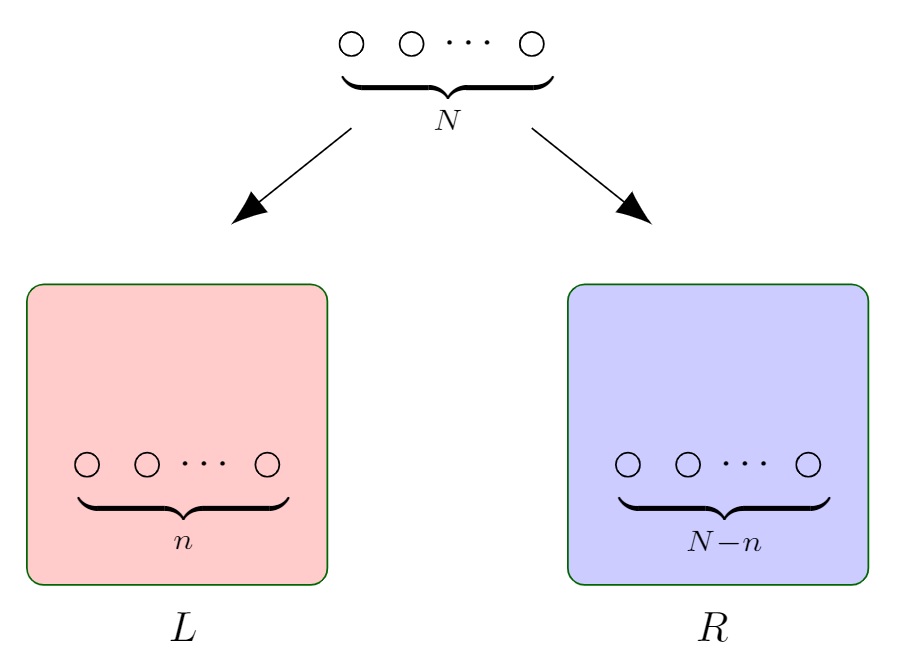}
	\caption{$N$ identical particles distributed in two distinctive subsystems $L$ and $R$. The two sets of identical particles ($n,N-n$) can have nontrivial correlation with nonzero spatial coherence~\cite{franco2018indistinguishability,chin2019entanglement}. } 
\label{split}
\end{figure}

Our bipartite entanglement of identical particles presupposes the division of the particles into two distinguished spatial modes (Fig.~\ref{split}). 
The expression of $n$ ($\le N$) identical particle partial wavefunctions in 1QL needs extra attention, for in this case we have no information on which of the $N$ particles are in which mode.
The partial wavefunction $|\P_1,\cdots,\P_n\>$ for the subsystem should be symmetrized with respect to the $N$ particle pseudolabels. The most general form that satisfies the exchange symmetry is given by
\begin{align}\label{partsym}
|\P_1,\cdots,\P_n\>^\pm &= \mathcal{N}_{[n]}(\P)\sum_{\substack{a_1< \cdots < a_n \\ =1}}^N e^{i\th_{a_1\cdots a_n}}\sum_{\s\in S_n}(\pm 1)^\s|\P_1\>_{A_{a_1}} \cdots |\P_{n}\>_{A_{a_n}}, 
\end{align} where $1\le a_p \le N$ for all $p=1,\cdots, n$ and $\mN_{[n]}$ is the normalization factor. Note that $\th_{a_1\cdots a_n}$ for each $\{a_1,\cdots, a_n\}$ does not affect the anti-symmetric property of the subsystem state, hence can be chosen arbitrarily~\footnote{This phase ambiguity implies a kind of gauge symmetry of 1QL to 2QL}.
All the phases $\th_{a_1\cdots a_n}$ can be set to zero for the case of bosons~\cite{chin2019entanglement}, which is however not true for fermions (see Appendix~\ref{31ferm} for a more detailed explanation).

With Eq.~\eqref{partsym}, we can define the partial trace of a given density marix $\r$ over a subsystem S with $n$ identical particles. 
Since the resultant reduced density matrix also preserves the pseudolabel exchange symmetry, the partial trace for identical particles in 1QL is dubbed \emph{symmetrized partial trace}~\cite{chin2019entanglement}.
Supposing $\{|\Phi_{1}^a, \cdots,\Phi_{n}^a\>\}_a$ composes the complete symmetric computational basis set of $n$ particles in the subsystem $S$, 
the identity matrix for $S$ is expressed with 
\begin{align}\label{IS}
\mathbb{I}_S = \sum_{n \in SSR}\sum_{a}|\Phi_{1}^a, \cdots,\Phi_{n}^a\> \<\Phi_{1}^a,\cdots, \Phi_{n}^a|
\end{align}
where the summation over $n$ means that we have to add the states for all possible $n$ that preserve the superselection rules (SSR, the particle number and parity SSR for bosons and fermions  respectively~\cite{wick1997intrinsic,wiseman2003entanglement,chin2020taming}).   
Then the reduced density matrix of the subsystem $\bar{S}$  by the symmetrized partial trace over $S$ is given by
\begin{align}\label{partial}
 \r_{\bar{S}} =	\Tr_S(\r) \equiv \sum_{n \in SSR}\sum_a \<\Phi_{1}^a, \dots,\Phi_{n}^a| \r |\Phi_{1}^a,\dots, \Phi_{n}^a\>.
\end{align}
With the reduced density matrix $\r_{\bar{S}}$, we can evaluate the amount of entanglement that is detectable and physical.
Appendix \ref{31ferm} provides an $(N,n)=(3,1)$ fermionic example. One can see that the extension of the argument to a mixed state case $\r= \sum_{a}\p_a|\P_a\>\<\P_a|$ ($\p_a \in \mathbb{R}^+$, $\sum_a \p_a = 1$) is straightforward .

\section{From 1QL to symmetric-exterior algebraic (SEA) approach }\label{sea}

The computations in Sec.~\ref{1ql} can be reproduced using the tensor algebra methods with suitable symmetries for identical particles. In this section, we provide a mathematical analysis on the relation of 1QL with the symmetric-exterior algebra (SEA).
We show that a bosonic wavefunction in 1QL corresponds to an elementary symmetric vector in the symmetric algebra, which can be straightforwardly extended to the relation between fermionic wavefunctions in 1QL and exterior vectors.




From a given element $x$ of a tensor algebra, one can impose an exchange symmetry with a proper mapping. Let  $\{\he_{1},\cdots,\he_{N}\}$ be a basis set of the $N$-dimensional vector space $V$ and $\{\he^{*1},\cdots,\he^{*N}\}$ its dual, so that the inner product among them is given by $(\he_{a}, \he^{*b}) = \d_a^b$.
Then we define the elementary symmetrizing map\footnote{It is named ``elementary'' symmetrizing map because it projects out all the tensor products of the same basis.} $S^{(k)}_e(x)$ of the element $x$ as follows: 
	\begin{definition}\label{elemsymm}
	For a $k$-th tensor power $T^kV$ of a vector space $V$ ($=\overbrace{V\otimes \cdots \otimes V}^{k} $) and a $k$-vector $x=\sum_{i_1,\cdots, i_k} x^{i_1\cdots i_k}\he_{i_1}\otimes \cdots \otimes \he_{i_k}$ in $T^kV$, the elementary symmetrizing map $S^{(k)}_e(x)$ is defined as
		\begin{align}
		S_e^{(k)}(x) = \sum_{i_1,\cdots, i_k} x^{i_1\cdots i_k} S_e^{(k)}(\he_{i_1}\otimes \cdots \otimes \he_{i_k}) 	\end{align}
		where
		\begin{align}
		S_e^{(k)}(\he_{i_1}\otimes \cdots \otimes \he_{i_k}) = \frac{1}{k!}\sum_{a_1, \cdots, a_k=1}^k|\eps_{i_{a_1}\cdots i_{a_k} }| (\he_{i_{a_1}}\otimes \cdots \otimes \he_{i_{a_k}})
		\end{align}	($|\eps_{i_{a_1},\cdots, i_{a_k}}|$ is the absolute value of the Levi-Civita symbol, which vanishes when $i_{a_k}=i_{a_l}$ for any $k$ and $l$).
	\end{definition}
	Note that the elementary symmetrizing map $S_e^{(k)}$ projects out any elements of $x$ that satisfy $i_a=i_b$ ($1\le a,b \le k$). 
	For example, when $x=\sum_{i,j=1}^N x^{ij}(\he_i\otimes \he_j)$, we have
	\begin{align}
	S_e^{(2)}(x) &= \sum_{i,j}x^{ij} S_e^{(2)}(\he_i\otimes \he_j)=\frac{1}{2} \sum_{i,j}x^{ij}|\eps_{ij}|(\he_i\otimes \he_j+\he_j\otimes \he_i) \nn \\ 
	&=\sum_{i\neq  j } \Big(\frac{x^{ij}+x^{ji} }{2}\Big)\Big(\frac{\he_i\otimes \he_j+\he_j\otimes \he_i}{2}\Big).
	\end{align} 
	The elementary symmetrizing map is simply denoted with the the elementary symmetric product $\vee$ as
	\begin{align}
	v\vee w \equiv S_e^{(k)}(v\otimes w)
	\end{align}
	where $v\in T^lV$ and $w \in T^{(k-l)}V$ ($0\le l\le k$).
	
Definition~\ref{elemsymm} gives a direct relation of a $N$ boson state in 1QL with the symmetric algebra. 
By defining $|\hP_i\>=\frac{1}{\sqrt{N}}\sum_{a=1}^N|\P_i\>_{A_a}\he^{*a}$, which is a vector in $V$ that is invariant under the exchange of pseudolabels, we can directly see that an $N$-boson total state $|\P\>^+$ (Eq.~\eqref{totalstate}) is related to an elementary symmetric $N$-vector as  (note that the equalities hold upto normalization)
	\begin{align}\label{bosontotal}
	|\P\>^+ 
	&= \sum_{\s\in S_N}|\P_1\>_{A_{\s(1)}} |\P_2\>_{A_{\s(2)}}\cdots |\P_N\>_{A_{\s(N)}} \nn \\
	&=\sum_{\substack{k_1<\cdots<k_N\\=1}}^N (\he_{k_1} \vee \cdots \vee \he_{k_N},|\hP_1\>\vee\cdots\vee |\hP_N\> ) \equiv F^+(|\hP_1\>\vee\cdots\vee |\hP_N\> ).
	\end{align} Here $(v,w)$ is the inner product between two $k$-vectors $v$ and $w$.
	And the $n$-boson subsystem state Eq.~\eqref{partsym} is written as  
	\begin{align}\label{bosonsubwf}
	&|\P_1,\cdots ,\P_n\>^+ \nn \\
	&=\sum_{\substack{a_1<\cdots< a_n\\=1}}^N e^{i\th_{a_1\cdots a_n}} \sum_{\s\in S_n} |\P_1\>_{A_{\s(a_1)}}\cdots |\P_n\>_{A_{\s(a_n)}} \nn \\
	&=\sum_{k_1<\cdots<k_n=1}^N e^{i\th_{k_1\cdots k_n}} (\he_{A_{k_1}}\vee \cdots \vee \he_{A_{k_n}}, |\hP_1\>\vee\cdots\vee |\hP_n\>).
	\end{align} 
Note that the phase ambiguity of Eq. \eqref{partsym} can be restricted to the complex coefficients of the $n$-dimensional basis vectors $ \he_{A_{k_1}}\vee \cdots \vee \he_{A_{k_n}}$ in Eq.~\eqref{bosonsubwf}. Considering the phase ambiguity is a purely mathematical feature that has no physical implication, one can state that  $|\hP_1\>\vee\cdots\vee |\hP_n\>$ ($1\le n \le N$) have the complete physical information on the total and sub- states for a set of identical bosons.  Hence, \emph{any subset of a given $N$ identical bosons can be exactly represented as an elementary symmetric product of single particle vectors $|\hP_i\>$ that is invariant under pseudolable exchanges.} From now on, we call such vectors ``elementary symmetric state vectors.''

	The connection of bosonic states to the exterior state vectors provides an intriguing insight on our physical system. 
	Since an elementary symmetric state vector $|\hP_{i_1}\>\vee \cdots \vee |\hP_{i_n}\>$ ($1 \le i_a\le N$, $^\forall a=1,\cdots , n$) is in the $n$th symmetric power $\bigvee^{n}(V)$ ($1\le n \le N$), all the possible subsets of identical fermions construct the complete set of the graded structure, i.e., $\{ |\hP_{i_1}\> \vee  \cdots\vee |\hP_{i_n} \>\}_{n=1}^N$ $\simeq$ $\oplus_{n=1}^N \bigvee^{n}(V)$. We can see that the projection between identical particle states corresponds to the shift in the graded structure. 
	Considering the bra states $\<\P_n,\cdots ,\P_1|$ ($\equiv (|\P_1,\cdots ,\P_n\>)^\dagger$) are expressed likewise with the symmetric products of $\<\hP_i| = \sum_{a}\<\P_i|^{A_a}e_{A_a}$, one can see that the projection of $|\Phi_1,\cdots ,\Phi_m\>$ onto $|\P_1,\cdots ,\P_n\>$ ($1\le n\le m\le N$) is replaced with the \emph{interior products} of $\<\hPh_n|\vee\cdots\vee \<\hPh_1|$ and $|\hP_1\>\vee\cdots\vee |\hP_m\>$.
	For example, when $n=1$,
	the interior product of $\<\hPh_1|$ and $|\hP_1\>\vee\cdots\vee |\hP_m\>$ is given by
	\begin{align}\label{bsn=1}
	&\<\hPh_1|\cdot |\hP_1\>\vee\cdots\vee |\hP_m\> \nn \\
	& = \sum_{j}\<\Phi_1|\P_j\>|\hP_1\>\vee\cdots\vee (|\hP_j\>)\vee |\hP_m\>
	\end{align}  (here $(|\hP_j\>)$ means that $|\hP_j\>$ is absent in the exterior state vector). For $n=2$, we have
	\begin{align}\label{bsn=2}
	&\<\hPh_1|\vee \<\hPh_2|\cdot |\hP_1\>\vee\cdots\vee |\hP_m\>\nn \\
	&= \<\hPh_1|\cdot (\<\hPh_2|\cdot |\hP_1\>\vee\cdots\vee |\hP_m\>) \nn \\
	&= \sum_{j}(-1)^{j-1}\<\Phi_2|\P_j\> (\<\hPh_1|\cdot |\hP_1\>\vee\cdots\vee (|\hP_j\>)\vee |\hP_m\>) \nn \\
	&= \sum_{j,k}(-1)^{j+k}\<\Phi_1|\P_j\>\<\Phi_2|\P_k\> \nn \\
	&\qquad  \quad \times
	\Big[|\hP_1\>\vee\cdots\vee (|\hP_j\>)\vee \cdots \vee(|\hP_j\>)\vee \cdots \vee |\hP_m\>\Big].
	\end{align} The same operation can be applied to an arbitrary $n$.

\begin{table}[t] \label{1st_to_sea}			
	\begin{center} 	
		\begin{tabular}{|  l | l |}
			\hline
			a set of identical particles & tensor algebra \\ \hline\hline 
			bosonic state & elementary symmetric  vector\\
			(fermionic state) & (exterior vector) \\ \hline
			symmetric partial trace & interior product  \\ \hline
			subsystems & graded structure 
			\\ \hline
		\end{tabular}\caption{The translation relation between 1QL and SEA for the description of identical particle entanglement}
	\end{center}

\end{table}
	
By replacing the state vectors of Eqs.~\eqref{IS} and \eqref{partial} with the corresponding elementary symmetric state vectors, we can calculate the reduced density matrix in a given subsystem with Eqs.~\eqref{bsn=1} and \eqref{bsn=2}. The algebraic relation of states defined in  NLA~\cite{compagno2018dealing} is equal to the interior product Eq.~\eqref{bsn=1}, which means that the physical states in NLA are operationally equivalent to the elementary symmetric vectors.
	
We can analyze the entanglement of fermions in the same way by replacing the elementary symmetric algebra with the exterior (antisymmetric) algebra~\cite{lam2015topics}.
The relation of the physical systems of identical particles with tensor algebras is summarized in Table 1. Appendix \ref{implicit} presents another way of describing identical particles using SEA with the implicit notation.

We can state that the tensor product $\otimes$ for the entanglement of non-identical particles is replaced with the elementary symmetric product $\vee$ and the exterior product $\wedge$ for bosons and fermions respectively. On the other hand, the entanglement of identical particles is a detector dependent quantity, which is determined by the spatial relation of particle wavefunctions to orthogonal detectors (which can be interpreted as coherence~\cite{chin2019entanglement}). Hence, the separability of a given state is not completely determined by the mathematical structure of the wavefunction itself. For a more advanced discussion on this issue, see Ref.~\cite{chin2020taming}. 


\section{From SEA to Algebraic approach}\label{algebraic}
	
Any state of identical particles has intrinsic correlations among single-particle spaces. 
	In other words, for a single particle Hilbert space $\mathcal{H}^{(1)}$,
	the anti-symmetrization (symmetrization) of fermionic (bosonic) wavefunctions sends the total Hilbert space $\mathcal{H} = \mathcal{H}^{(1)\otimes N}$  to $\mathcal{H} = \bigwedge^N \mathcal{H}^{(1)}$ ($\mathcal{H} = \bigvee^N \mathcal{H}^{(1)}$).
	Since the total Hilbert space is invariant under the action of the algebra $\mathcal{A}$ of observables, the observables  also must be invariant under the symmetrizations. Therefore, the partial trace defined as the formalism of non-identical particles is no more valid. It was the motivation of Ref.~\cite{chin2019entanglement} to introduce the symmetrized partial trace, and also the motivation of Refs. \cite{balachandran2013entanglement,balachandran2013algebraic} to suggest the concept of \emph{restrictions to subalgebras}  as the replacement of particle trace. Therefore, it is natural to ask about the relation between the symmetrized partial trace and subalgebra restriction.
	Indeed, one can show that \emph{the restriction to subalgebras is equivalent to the symmetrized partial trace for the case of identical particles.}
	
	Instead of a Hilbert space $\mH$ and linear operators acting on it, quantum systems can be described with an abstract algebra of physical observable, i.e., $C^*$ algebra, in which the algebra $\mathcal{A}$ and state $\w$ describe a given quantum system. By Gel'fand, Naimark, and Segal (GNS)~\cite{gelfand1943imbedding,segal1947irreducible} construction, the data ($\mathcal{A}, \w$) can reconstruct the corresponding Hilbert space $\mH_\w$. Here $\w$ is a linear map from the given algebra $\mathcal{A}$ to $\mathbb{C}$ that corresponds to the particle state $\r$ that maps an observable to a complex number with trace.
	The relation between Hilbert space and GNS representation of quantum physics is listed in Table 2. 		For a more thorough explanation on GNS construction, see Refs.~\cite{balachandran2013algebraic,fewster2020algebraic}.

In the GNS construction, the notion of partial trace can be replaced with the \emph{restriction} $\w_0 := \w|_{\mathcal{A}_0}$ of a state $\w$ on $\mathcal{A}$ to a subalgebra $\mathcal{A}_0$. 
	Suppose $\w$ is represented as a density matrix $\r_\w$, i.e.,
	$\w(\a) = \Tr (\r_\w\a)$ ($\a \in \mathcal{A}$).
	Then for a subalgebra $\mathcal{A}_0$ of $\mathcal{A}$, we can define a \emph{restriction of $\w$ to $\mathcal{A}_0$} as a state $\w|_{\mathcal{A}_0}$: $\mathcal{A}_0 \to \mathbb{C}$, i.e., from $\a_0$ ($\in \mathcal{A}_0$) to  $\w|_{\mathcal{A}_0}(\a_0)$ so that
	\begin{align}
	\w|_{\mathcal{A}_0}(\a_0)= \w(\a_0)
	\end{align} holds \cite{balachandran2013entanglement}. A simple example is a bipartite system of non-identical particles $A$ and $B$ in the Hilbert space $\mH =\mH_A\otimes\mH_B$. For a given vector state $|\P\> (\in  \mH)$ and the subalgebra $\mathcal{A}_0 = \{ \a_0\in\mathcal{A}| \a_0=K_A\otimes \bI_B \}$ ($K_A$ is an observable on $\mH_A$), we can see that the following equality holds:
	\begin{align}
	\w|_{\mathcal{A}_0}(\a_0) = \Tr_{\mH_A}(\r_AK_A) = \Tr_{\mH}(\r \a_0) =  \w(\a).
	\end{align} where $\r_A \equiv \Tr_{\mH_B}(\r)$. For this special case, the reduced density matrix $\r_A$ is equivalent to the restriction of $\r$ to $\mathcal{A}_0$. 

	\begin{table}[t]
	\begin{center}
		\begin{tabular}{|  l | l | l |}
			\hline
			$ $  & Hilbert space & GNS \\ \hline\hline 
			observables & $\mathcal{O} (=\mathcal{O}^\dagger)$ & $\a \in \mathcal{A}$  \\ \hline 
			state & $\r$ ($\Tr\r=1$ and $\r\ge 0$) & $\w: \mathcal{A} \overset{\w} \longrightarrow \mathbb{C}$  \\ \hline
			expectation value & $\<\mathcal{O}\>_\r =  \Tr(\r\mathcal{O}) \in \mathbb{C}$ & $\w(\a)\in\mathbb{C}$ \\ \hline  
		\end{tabular}	
	\end{center} \caption{Understading Hilbert space quantum system with GNS construction}
	\end{table}
	
	However, when particles are identical, a single particle observable is  not included in a single Hilbert space $\mH^{(1)}$ in the total Hilbert space $\mH$. Along  with the permutation symmetry of identical particles, the observables also must be invariant under the permutations of single Hilbert spaces. 
	As pointed out in Refs.~\cite{balachandran2013entanglement, balachandran2013algebraic}, this permutation symmetry of observables causes the discrepancy between the subalgebra restriction and the traditional form of partial trace that is valid for non-identical particles. Hence the authors of Refs.~\cite{balachandran2013entanglement, balachandran2013algebraic} claimed that (instead of the partial trace) the restriction of states to subalgebras is appropriate for analyzing the entanglement of identical particles. 
	
	On the other hand, we can see that the newly defined  symmetrized partial trace of identical particles plays the role of subalgebra restriction.
	\begin{theorem} The symmmetrized partial trace for identical particles maps  a density matrix $\r$ and a state $\w$ on the algebra $\mathcal{A}$ to the restrictions $\r_0$ and $\w_0$ on the subalgebra $\mathcal{A}_0$ defined in a subsystem $S$, i.e.,
		\begin{align}\label{sptrest}
		\w|_{\mathcal{A}_0}(\a_0) = \mathcal{TR}_{S}(\r_0 \a_0) = \mathcal{TR}(\r \a_0).
		\end{align}
		
	\end{theorem}
	\begin{proof} First, we will show that Eq.~\eqref{sptrest} holds for the simplest case, i.e., $(N,n)=(2,1)$. Suppose a subalgebra $\mathcal{A}_0$ divides the wave fuctions $\P_{i}$ into two parts, i.e., $\P_a$ and $\P_{\m}$ (states with latin indices form a complete orthonormal basis set of the subsystem $L$ and those with greek indices form a complete orthonormal basis set of the subsystem $R$. $\P_a$ are rotated to each other by operators in $\mathcal{A}_0$).
		Then an observable $\a_0$ ($\in \mathcal{A}_0$) for 2-particle states in $\mathcal{A}$ is written as
		\begin{align}
		\a_0 &= \sum_{a,b,\m} \a_{ab} (|\hP_a\>\wedge|\hP_\m\>)(\<\hP_b|\wedge\<\hP_\m|)   \nn \\ 
		&\equiv \sum_{a,b} \a_{ab}|\hP_a\>\<\hP_b|\wedge \hat{\bI}_R.
		\end{align}
		For a given state $|\P\> = \sum_{c,\m} c_{c\m} |\P_c\>\wedge |\P_\m \>$, we have
		\begin{align}
		&\mathcal{TR}_{R}|\P\>\<\P|\nn \\
		& = \sum_{c,d,\m,\n,\w}c_{c\m}c^*_{\n d} \< \P_\w|\cdot |\hP_c\>\wedge |\hP_\m\>\<\hP_\n|\wedge \< \hP_d|\cdot|\hP_\w\>  \nn \\
		&=\sum_{c,d}(\sum_\w c_{c\w}c^*_{\w d}) |\hP_c\>\<\hP_d|  \equiv \sum_{c,d} X_{cd}|\hP_c\>\<\hP_d| \equiv \r_L.
		\end{align} 
		Then
		\begin{align}
		&\mathcal{TR}_L (\r_L\a_0) \nn \\
		&= \sum_{a,b,c,d}\a_{ab}X_{cd} \<\P_d|\P_a\>\<\P_b|\P_c\>= \sum_{ab} \a_{ab}X_{ba} 
		\end{align}
		On the other hand,
		\begin{align}
		\mathcal{TR}(\a_0|\P\>\<\P|) =& \sum_{\substack{a,b,\m\\c,\n,d,\w}}\a_{ab}c_{c\n}c^*_{\w d} \<\hP_\w|\wedge \< \hP_d|\cdot |\hP_a\>\wedge|\hP_\m\>\nn \\
		&\qquad\quad   \times \<\hP_\m|\wedge\<\hP_b|\cdot |\hP_c\>\wedge |\hP_\n\> \nn \\
		=& \sum_{a,b}\a_{ab}X_{ba},	
		\end{align} hence Eq.~\eqref{sptrest} holds for $(N,n)=(2,1)$ case. The extension of this method to the general $(N,n)$ case is straightforward.
	\end{proof}
	In the above proof, one can see that possible states for a given subsystem (which is determined by the implementation of detectors) determines the selection of subalgebra in GNS representation. This means that \emph{the detector dependence of entanglement is equivalent to the subalgebra dependence of entanglement mentioned in Ref.~\cite{balachandran2013algebraic}.} This point is rigorously confirmed in Ref.~\cite{chin2020taming} by proving that the total Hilbert space is factorizable according to the locality of subsystems. 
	
	Since the computation of symmetrized partial trace is a straightforward process, it is more convenient to obtain the same reduced density matrix for a set of identical particles using the method of symmetrized partial trace than using the subalgebra restriction technic. Nonetheless, the subalgebra restiction is still a useful tool to understand some multipartite systems that are less symmetrical than a set of identical particles.
	

	\section{Conclusions}\label{conclusions}
	We have discussed the reduced density matrix of nonlocal identical particles in three types of formalisms, i.e., the symmetrized partial trace in the first quantization language (1QL), the interior products in symmetric and exterior algebra (SEA, which correspond to the no-labeling approach), and the subalgebra restriction in the GNS representation. Our current work bridges different viewpoints that use diverse languages to understand the quantum correlation of identical particles. We also expect that it provides a tool to understand the nonlocality of identical particle systems such as quantum field theory (e.g., Tomita-Takesaki theory \cite{tomita1967canonical,takesaki2006tomita,witten2018aps} and symmetric product orbifold CFT \cite{balasubramanian2019entanglement}).

	\section*{Acknowledgements}
	This work is supported by Basic Science Research Program through the National Research Foundation of Korea (NRF) funded by the Ministry of Education, Science and Technology (NRF-2015R1A6A3A04059773, NRF-2019R1I1A1A01059964, NRF-2019M3E4A1080227, NRF-2019M3E4A1079666). JH acknowledges the support by the POSCO Science Fellowship of POSCO TJ Park Foundation. 
	

\appendix

\section{Identical particles in the implicit notation}\label{implicit}
We follow Ref.~\cite{chin2020taming} for the implicit description of identical particle.
With Eq.~\eqref{nnonidordered}, the wavefunction of $N$ identical particles in the implicit notation is expressed as the following definition.


\begin{definition} An $N$ boson state is defined with  
the symmetric tensor product $\vee$ as
\begin{align}\label{bosons}
 |\P_1\>\vee\cdots \vee |\P_N\> = \frac{1}{\mN} \sum_{\s\in S_N} |\P_{\s(1)}\>\otimes \cdots \otimes |\P_{\s(N)}\>.
\end{align} An $N$ fermion state is defined with the exterior (or antisymmetric) tensor product $\wedge$ as
\begin{align}\label{fermion}
 |\P_1\>\wedge\cdots \wedge |\P_N\> = \frac{1}{\mN} \sum_{\s\in S_N} (-1)^\s |\P_{\s(1)}\>\otimes \cdots \otimes |\P_{\s(N)}\>
\end{align} where $(-1)^\s$ is the signature of $\s$. 
\end{definition}

We use $\otimes_{\pm}$ for both $\vee$ and $\wedge$ when the tensor product can be any of them.  
The transition amplitude from $|\P_1 \>\otimes_{\pm} \cdots \otimes_{\pm} |\P_N\>$ to $|\Phi_1 \>\otimes_{\pm} \cdots \otimes_{\pm} |\Phi_N\>$ is calculated with the scalar product of two tensors, 
\begin{align}
 \<\Phi_N|\vee \cdots \vee \<\Phi_1| \cdot |\P_1 \>\vee \cdots \vee |\P_N\> = \frac{1}{\mN^2}Per[\<\Phi_i|\P_j\>], \nn \\
  \<\Phi_N|\wedge \cdots \wedge \<\Phi_1| \cdot |\P_1 \>\wedge \cdots \wedge |\P_N\> = \frac{1}{\mN^2}Det[\<\Phi_i|\P_j\>],
\end{align}
where $Per$ and $Det$ are the permanent and determinant of a $N\times N$ matrix that has $\<\Phi_i|\P_j\>$ as its elements.

Then the creation operator $\ha^\dagger(\P)$ is defined as
  \begin{align}\label{creat}
      \ha^\dagger (\P) (|\P_1\>\otimes_\pm\cdots\otimes_\pm  |\P_N\>) = |\P\>\otimes_\pm|\P_1\>\otimes_\pm\cdots \otimes_\pm  |\P_N\> 
  \end{align} and the annihilation operator is defined as
  \begin{align}\label{annih}
      \ha (\P) (|\P_1\>\otimes_\pm\cdots\otimes_\pm |\P_N\>) 
     &\equiv \<\P|\cdot |\P_1\>\otimes_\pm \cdots\otimes_\pm |\P_N\> \nn \\
      &=\sum_{i=1}^N (\pm 1)^{i-1} \<\P|\P_i\>|\P_1\>\otimes_\pm \cdots \otimes_\pm (|\P_i\>) \otimes_\pm\cdots \otimes_\pm |\P_N\>,
  \end{align}  where $(|\P_i\>)$ in the last line denotes that a particle  of pseudolabel $A_i$ in state $|\P_i\>$ is absent. It is direct to see that Eq.~\eqref{annih} corresponds to Eq.~\eqref{bsn=1}. This relation reveals that an identical particle state in the implicit notation is identical to the elementary symmetric state vectors. Indeed, the phase ambiguity in the explicit notation does not appear in the implicit notation such as in SEA formalism. In this sense, we can state that \emph{the implicit description of identical particles with the exchange symmetries is another form of SEA.}


\section{$N=3$ fermion example}\label{31ferm}

Here we compute $(N,n)=(3,1)$ fermion example, i.e., one fermion at the left detector and two in the right (Fig.~\ref{split}). Let the fermions have four internal degrees of freedom. Then the possible states are expressed as 
 	\begin{align}
	(|\P_1\>,|\P_2\>,|\P_3\>,|\P_4\>) = (|L,0\>,|L,1\>,|L,2\>,|L,3\> ), \nn \\
	(|\P_5\>,|\P_6\>,|\P_7\>,|\P_8\>) = (|R,0\>,|R,1\>,|R,2\>,|R,3\> )
	\end{align} 
	and the identity operator of the subsystem $R$ is given by
	\begin{align}
	\bI_R&=  |R,0\>\<R,0| + |R,1\>\<R,1|+ |R,2\>\<R,2| + |R,3\>\<R,3| \nn \\
	&=  |\P_5\>\<\P_5| + |\P_6\>\<\P_6|+ |\P_7\>\<\P_7| +|\P_8\>\<\P_8| 
	\end{align}
	where 
	\begin{align}\label{ferm31}
	     	|\P_i\>^f= \frac{1}{\sqrt{3}}(|\P_i\>_{A}+ e^{i\th_i}|\P_i\>_{B}+e^{i\eta_i}|\P_i\>_{C})
	\end{align}
 for $i=5,6,7,8$ (Eq.~\eqref{partsym}). Suppose the $N=3$ fermion state is as follows: 
	\begin{align}
	|\P\>^- &= \frac{1}{\sqrt{2}}(|\P_1,\P_2,\P_5\>^- + |\P_1,\P_3,\P_6\>^-).
	\end{align}
	We can compute the reduced density matrix $\r_L$ for the subsystem $L$ by using the symmetrized partial trace technic as
		\begin{align}
		\r_L= \mathcal{TR}_R|\P\>\<\P|^- =&\Big( \<\P_5|\P_1,\P_2,\P_5\>^-\<\P_5,\P_2,\P_1|\P_5\>^- +  \<\P_6|\P_1,\P_3,\P_6\>^-\<\P_6,\P_3,\P_1|\P_6\>^-\Big). 
		\end{align}
		Since 
		\begin{align}\label{ferm32}
		&\<\P_5|\P_1,\P_2,\P_5\>^- \nn \\
 		&=\Big[ (|\P_1\>_B|\P_2\>_C - |\P_2\>_B|\P_1\>_C ) +e^{-i(\th_4+\pi)}(|\P_1\>_A|\P_2\>_C - |\P_2\>_A|\P_1\>_C) +e^{-i\eta_4}( |\P_1\>_A|\P_2\>_B - |\P_2\>_A|\P_1\>_B) \Big] \nn \\
		&\equiv |\P_1,\P_2\>^-
		\end{align} (note that the second line of the above equation fits Eq.~\eqref{partsym} with $n=2$)
	and so on,  we have 
	\begin{align}
	\r_L = \frac{1}{2}(|\P_1,\P_2\>\<\P_1,\P_2|^- + |\P_1,\P_3\>\<\P_1,\P_3|^-).
	\end{align} This state is entangled with the von Neumann entropy $S(|\P\>)=1$. The phase change from Eq.~\eqref{ferm31} to the second line of Eq.~\eqref{ferm32} shows that the relative phase  cannot be fixed for the fermionic case. 
	
	The same computation is possible in the exterior vector formalism by expressing 
	\begin{align}
	|\P\>^- &= \frac{1}{\sqrt{2}}(|\hat{\P}_1\>\wedge|\hat{\P}_2\>\wedge |\hat{\P}_5\> + |\hat{\P}_1\>\wedge|\hat{\P}_3\>\wedge|\hat{\P}_6\>).
	\end{align} 
	Then using Eq.~\eqref{bsn=1}, the reduced density matrix is given by
	\begin{align}
	\r_L = \frac{1}{2}\big(|\hat{\P}_1\>\wedge|\hat{\P}_2\>\<\hat{\P}_2|\wedge\<\hat{\P}_1| + |\hat{\P}_1\>\wedge|\hat{\P}_3\>\<\hat{\P}_3|\wedge\<\hat{\P}_1|\big).
	\end{align}

\bibliographystyle{unsrt}
\bibliography{extgns}

\end{document}